\documentclass[a4paper,UKenglish,cleveref, autoref, thm-restate]{lipics-v2021}
\usepackage{amssymb,amsmath, amsfonts}
\usepackage{mathtools}

\usepackage{xcolor}
\usepackage{xspace}
\usepackage{tikz}
\usetikzlibrary{patterns,arrows,backgrounds,positioning,calc,fit,shapes.geometric,decorations.pathreplacing,decorations.markings, tikzmark,topaths}
\usepackage{tkz-euclide}

\tikzstyle{internal} = [draw, fill, shape=circle, text=black, inner sep=4pt]
\tikzstyle{external} = [shape=circle, draw]
\tikzstyle{square}   = [draw, fill, rectangle, inner sep=5pt]
\tikzset{lnode/.style = {
    circle, 
    draw=cyan!30!black, 
    thick,
    inner sep=1.5pt,
    minimum size=15pt } }
\tikzset{uedge/.style = {
    draw=cyan!20!black, 
    very thick} }
\usepackage{fancybox}
\usepackage[ruled]{algorithm}
\usepackage{algpseudocode}
\usepackage[nocompress]{cite}
\usepackage{booktabs}

\DeclareMathOperator\comp{\mathcal{K}}

\DeclareMathOperator*{\argmin}{arg\,min}

\newcommand{\bR}{\mathbb{R}}

\newcommand{\newfontobj}[2]{
\newcommand{#1}[1]{
    \expandafter\def\csname##1\endcsname{{#2 ##1}}}}

\newfontobj{\class}{\rm} % Typeset Classes in roman font

\class{PSPACE}
\class{L}
\class{SL}
\class{BPL}
\class{RL}
\class{NC}
\class{ZPL}
\class{NPSPACE}
\class{ASPACE}
\class{NL}
\class{coNL}
\class{EXP}
\class{NEXP}
\class{coNEXP}
\class{NE}
\class{E}
\class{AM}
\class{MA}
\class{NP}
\class{DNP}
\class{UP}
\class{P}
\class{RP}
\class{BPP}
\class{ZPP}
\class{EXPSPACE}
\class{coNP}
\class{coRP}
\class{coAM}
\class{IP}
\class{PCP}
\class{AC}
\class{TC}
\class{NC}
\class{MIP}

\class{BP}

\newcommand{\hitting}{\textsc{Hitting Set}\xspace}
\newcommand{\mwc}{\textsc{Multiway Cut}\xspace}
\newcommand{\mc}{\textsc{Multicut}\xspace}
\newcommand{\uml}{\textsc{Uniform Metric Labeling}\xspace}
\newcommand{\kcut}{$k$\textsc{-Cut}\xspace}
\newcommand{\lifted}{\textsc{Lifted Cut}\xspace}

\newcommand{\TiTj}{\ensuremath{T_i-T_j}\xspace}
\newcommand{\nTiTj}{\textsc{All-to-All}\xspace}
\newcommand{\tiTj}{\ensuremath{t_i-T_j}\xspace}
\newcommand{\ntiTj}{\textsc{Single-to-All}\xspace}
\newcommand{\titj}{\ensuremath{t_i-t_j}\xspace}
\newcommand{\ntitj}{\textsc{Single-to-Single}\xspace}
\newcommand{\tijTj}{\ensuremath{t_i^j-T_j}\xspace}
\newcommand{\ntijTj}{\textsc{Some-to-All}\xspace}
\newcommand{\tijtj}{\ensuremath{t_i^j-t_j}\xspace}
\newcommand{\ntijtj}{\textsc{Some-to-Single}\xspace}
\newcommand{\tijtji}{\ensuremath{t_i^j-t_j^i}\xspace}
\newcommand{\ntijtji}{\textsc{Some-to-Some}\xspace}
\newcommand{\stj}{\ensuremath{s-t_j}\xspace}
\newcommand{\nstj}{\textsc{Fixed-to-Single}\xspace}

\title{Multiway Cuts with a Choice of Representatives} %TODO Please add

\author{Kristóf Bérczi}{MTA-ELTE Matroid Optimization Research Group and HUN-REN{--}ELTE Egerváry Research Group on Combinatorial Optimization, Department of Operations Research, Eötvös Loránd University, Budapest, Hungary}{kristof.berczi@ttk.elte.hu}{https://orcid.org/0000-0003-0457-4573}{}
\author{Tamás Király}{HUN-REN{--}ELTE Egerváry Research Group on Combinatorial Optimization, Department of Operations Research, Eötvös Loránd University, Budapest, Hungary}{tamas.kiraly@ttk.elte.hu}{[https://orcid.org/0000-0001-7218-2112]}{}
\author{Daniel P. Szabo}{Department of Operations Research, Eötvös Loránd University, Budapest, Hungary}{dszabo2@wisc.edu}{https://orcid.org/0009-0009-7263-1614}{}
\authorrunning{K. Bérczi, T. Király, and D. P. Szabo} %TODO mandatory. First: Use abbreviated first/middle names. Second (only in severe cases): Use first author plus 'et al.'

\Copyright{Kristóf Bérczi, Tamás Király, and Daniel P. Szabo} %TODO mandatory, please use full first names. LIPIcs license is "CC-BY";  http://creativecommons.org/licenses/by/3.0/

 %TODO mandatory; please add comma-separated list of keywords

\funding{
The research was supported by the Lend\"ulet Programme of the Hungarian Academy of Sciences -- grant number LP2021-1/2021, by the Ministry of Innovation and Technology of Hungary -- grant number ELTE TKP 2021-NKTA-62, and by Dynasnet European Research Council Synergy project -- grant number ERC-2018-SYG 810115.
}

\nolinenumbers %uncomment to disable line numbering

\EventEditors{Rastislav Kr\'{a}lovi\v{c} and Anton\'{i}n Ku\v{c}era}
\EventNoEds{2}
\EventLongTitle{49th International Symposium on Mathematical Foundations of Computer Science (MFCS 2024)}
\EventShortTitle{MFCS 2024}
\EventAcronym{MFCS}
\EventYear{2024}
\EventDate{August 26--30, 2024}
\EventLocation{Bratislava, Slovakia}
\EventLogo{}
\SeriesVolume{306}
\ArticleNo{18}
\begin{document}

\maketitle

\begin{abstract}
In the \mwc problem, we are given an undirected graph with nonnegative edge weights and a subset of $k$ terminals, and the goal is to determine a set of edges of minimum total weight whose removal disconnects each terminal from the rest. 
The problem is APX-hard for $k\geq 3$, and an extensive line of research has concentrated on closing the gap between the best upper and lower bounds for approximability and inapproximability, respectively.

In this paper, we study several generalizations of \mwc where the terminals can be chosen as \emph{representatives} from sets of \emph{candidates} $T_1,\ldots,T_q$. 
In this setting, one is allowed to choose these representatives so that the minimum-weight cut separating these sets \emph{via their representatives} is as small as possible. 
We distinguish different cases depending on (A) whether the representative of a candidate set has to be separated from the other candidate sets completely or only from the representatives, and (B) whether there is a single representative for each candidate set or the choice of representative is independent for each pair of candidate sets.

For fixed $q$, we give approximation algorithms for each of these problems that match the best known approximation guarantee for \mwc. Our technical contribution is a new extension of the CKR relaxation that preserves approximation guarantees.
For general $q$, we show $o(\log q)$-inapproximability for all cases where the choice of representatives may depend on the pair of candidate sets, as well as for the case where the goal is to separate a fixed node from a single representative from each candidate set. As a positive result, we give a $2$-approximation algorithm for the case where we need to choose a single representative from each candidate set. 
This is a generalization of the $(2-2/k)$-approximation for \kcut, and we can solve it by relating the tree case to optimization over a gammoid.

% \keywords{Approximation algorithms, CKR relaxation, Combinatorial optimization, Multiway cut}  
\end{abstract}

%%%%%%%%%%%%%%%%
\section{Introduction}
\label{sec:intro}
%%%%%%%%%%%%%%%%

For an undirected graph $ G=(V,E) $ with weight function $w:E\to \bR_+$, the \mwc problem asks for a minimum-weight cut $ C\subseteq E $ separating any pair of terminals in a given terminal set $ S=\{s_1,\dots s_k\} $. As cuts can be identified with partitions of the nodes, this is equivalent to finding a node coloring of $ G $ with $ k $ colors such that terminal $ s_i $ is colored with color $ i $ for $i\in[k]$, and we seek to minimize the total weight of dichromatic edges.
%\vspace{-.4cm}

\subsection{Previous work}
Dahlhaus et al. \cite{Dahlhausetal} showed that \mwc is NP-hard even for $ k=3 $, and 
%It remains, however, a significant problem for approximation algorithms. A classical result is the 
provided a very simple combinatorial $ (2-2/k) $-approximation that works as follows. For each $s_i$, determine a minimum-weight cut $C_i\subseteq E$ that separates $s_i$ from all other $s_j$ for $j\neq i$ -- such a cut is called an \emph{isolating} cut of $s_i$ -- and then take the union of the $ k-1 $ smallest ones among the $k$ cuts thus obtained. In an optimal multiway cut, the boundary of the component containing $s_i$ is a cut isolating $s_i$, hence its weight is at least as large as that of $C_i$. Summing up these inequalities for all but the largest isolating cuts, since this counts each edge at most twice except for the boundary of the largest one, leads to a $ (2-2/k) $-approximation.

Since the pioneering work of Dahlhaus et al., \mwc has been  a central problem in combinatorial optimization. The best known approximability as well as inapproximability bounds are based on a geometric relaxation called the \emph{CKR relaxation}, introduced by C\v{a}linescu, Karloff and Rabani~\cite{CKR}. The current best approximation algorithm is due to Sharma and Vondr\'ak~\cite{SharmaVondrak} with an approximation factor of $ 1.2965 $, while the best known lower bound (assuming the Unique Games Conjecture) is slightly above $ 1.2 $ \cite{IntGap}.

Various generalizations of \mwc have been introduced. In the \mc problem, we are given an undirected graph with non-negative edge weights, together with a demand graph consisting of edges $ (s_1,t_1),\ldots, (s_k,t_k)$, and the goal is to determine a minimum-weight cut whose removal disconnects each $ s_i $ from its pair $ t_i $. \mc is NP-hard to approximate within any constant factor assuming the Unique Games Conjecture \cite{chawla2006hardness}, and there is a polynomial-time $O(\log k)$-approximation algorithm \cite{garg1996approximate}. The \uml problem takes as input a list of possible colors for each node in an edge-weighted graph, and asks for a coloring that respects these lists with the minimum total weight of dichromatic edges; %The \uml problem generalizes \mwc. Indeed, 
\mwc arises as a special case when the terminals have distinct lists of length $ 1 $ and all other nodes can be colored arbitrarily. Kleinberg and Tardos~\cite{KleinbergTardos} gave a 2-approximation to \uml with a tight integrality gap using a geometric relaxation, similar to that of CKR. 
%Finally, a problem that is closely related to \mwc is \kcut. 
In the \kcut problem, we are given only an edge-weighted graph $G$ and a positive integer $k$, and the goal is to find a minimum-weight cut whose deletion breaks the graph into $k$ components. One can think of this problem as a version of \mwc where the terminals can be chosen freely. The \kcut problem admits a $2$-approximation~\cite{kcut2apx} that is tight~\cite{kcut}. 
% \tnote{Add here the known results about uniform metric labeling.}
%Another generalization is 
The \textsc{Steiner Multicut} \cite{Steiner} problem takes as input an undirected graph $G$ and subsets $X_1,X_2,\ldots, X_q$ of nodes, and asks for a minimum cut such that each $X_i$ is separated into at least $2$ components. A generalization of \textsc{Steiner Multicut} is the \textsc{Requirement Cut} problem\cite{RequirementCut}, where requirements $r_i$ are given for each set $X_i$, and the goal is to find the minimum cut that cuts each $X_i$ into at least $r_i$ components. The current best algorithms for \textsc{Requirement Cut} are those given in \cite{RequirementCut, Req2}, of which we will use the $O(\log k \log q)$ approximation, where $k=|\bigcup_{i=1}^q X_i|\leq n$.
%\vspace{-.4cm}

\subsection{Our results}
We introduce generalizations of \mwc, where we are allowed to choose \emph{representatives} from some terminal candidate sets $T_1,\ldots,T_q\subseteq V$, and the goal is to find the minimum-weight cut separating these sets \emph{via their representatives}. The variants are distinguished by (A) whether the representative has to be separated from all candidates of the other candidate sets or only from their representatives, and (B) whether there is a single representative for each candidate set or whether the choice of representative is independent for each pair of candidate sets.
In order to make it easier to distinguish these problems, we use the following naming rules. 
%The precise definitions, as well as exact conditions for feasibility, are given in Appendix~\ref{app:defs}. 

\begin{itemize}
    \item When the goal is to separate \emph{all} candidates, we use \textsc{All}; for example, the \nTiTj problem requires all nodes of $T_i$ to be separated from all nodes of $T_j$, for each $i \neq j$. 
    \item When the goal is to choose a \emph{single} representative for each candidate set, we use \textsc{Single}, and we denote the chosen representative of $T_i$ by $t_i$. For example, the \ntitj problem requires choosing a representative $t_i \in T_i$ for every $i\in [q]$, and finding a cut that separates $t_i$ from $t_j$ for all $i \neq j$. On the other hand, \ntiTj requires the chosen representative $t_i \in T_i$ to be separated from every node of $T_j$, for all $i \neq j$.
    
    \item When only \emph{some} representative of $T_i$ ought to be separated from some part of $T_j$ for each $i,j$ pair, we use \textsc{Some}, and denote the representative chosen from $T_i$ to be separated from $T_j$ by $t^j_i$. For example, the {\ntijtji} problem asks for a minimum-weight subset of edges such that after deleting these edges, for any pair $i\neq j$, there are nodes $t^j_i\in T_i$ and $t^i_j\in T_j$ that are in different components.

    \item When there is a \emph{fixed} node that needs to be separated from the candidate sets, we  use \textsc{Fixed}, and denote the fixed node by $s$. In the \nstj problem, we are given a fixed node $s$, and we want a minimum-weight subset of edges such that after deleting these edges, $s$ is separated from at least one element $t_j \in T_j$ for every $j\in [q]$.
\end{itemize}

These problems are natural generalizations of \mwc that provide various ways to interpolate between problems with fixed terminals like \mwc and problems with freely chosen terminals like \kcut. Although, as we will discuss later, some of our problems are equivalent or closely related to problems that have already been considered in the literature, a systematic study of this type of generalization has not yet been done, and some of our results (Theorem \ref{thm:lift}, Theorem \ref{thm:ti-tj}) require new observations and techniques.

In each problem, we want to minimize over all possible choices of representatives, as well as over all possible subsets of edges. The problem where we need to separate each candidate set from every other, {\nTiTj}, is equivalent to \mwc by contracting each candidate set to a single node. The other problems are not directly reducible to \mwc.  
We denote by $ \alpha\approx 1.2965 $ the current best approximation factor for \mwc \cite{SharmaVondrak}. The different problems, as well as our results, are summarized in Table \ref{tab}. The main results that require new techniques are indicated in bold in the table, and are discussed in the next subsection.
\begin{table}[!tbh]
\renewcommand{\arraystretch}{1.2} 
\begin{center}
    \begin{tabular}
        {@{}c@{\quad}c@{\quad}r@{\quad}r@{}}
        Problem&Demands & Fixed $ q $ & Unbounded $ q $ \\
        \toprule
        {\nTiTj}&{\TiTj}&$ \alpha $-approx&$ \alpha $-approx\\
        %							\hline
        {\ntiTj}&{\tiTj}&\textbf{$ \alpha $-approx}&$ 2 $-approx\\
        %							\hline
        {\ntitj}&{\titj}&$ \alpha $-approx&\textbf{Tight $ 2 $-approx}\\
        %							\hline
        {\nstj}&{\stj}&In P&No $ o(\log q) $ approx\\
        %							\hline
        {\ntijtj}&{\tijtj}&$ \alpha $-approx&No $ o(\log q) $ approx\\
        %							\hline $ o(\log q) $ approx
        {\ntijtji}&{\tijtji}&$ \alpha $-approx&$ O(\log q \cdot \log n) $ approx \cite{RequirementCut}\\
        %							\hline
        {\ntijTj}&{\tijTj}&\textbf{$ \alpha $-approx}&No  $ o(\log q) $ approx\\[2pt]\\
        %							\hline
    \end{tabular}
    \caption{A summary of our results, where $ \alpha\approx 1.2965 $ \cite{SharmaVondrak} is the current best approximation factor for \mwc. The tightness of $2$-approximation assumes SSEH, 
    %(defined in Appendix \ref{app:conj})
    while the other inapproximability results hold assuming $\mathrm{P}\neq \mathrm{NP}$. The main results are highlighted in bold.}\label{tab}
    \vspace{-1cm}
\end{center}
\end{table}

\subsection{Techniques}

\subparagraph*{Approximation when $q$ is part of the input.} We give $ 2 $-approximations for {\ntiTj} and {\ntitj}. For the latter, we first give an exact algorithm on trees, by showing that the feasible solutions have a gammoid structure. This then leads to a $ 2 $-approximation for general graphs using the Gomory-Hu tree, which is best possible, since {\ntitj} generalizes the {\kcut} problem. Also, we show that the {\ntijtji} problem is equivalent to \textsc{Steiner Multicut}, leading to an $O(\log q \cdot \log n)$  approximation in this case.

\subparagraph*{Approximation for fixed $q$.} Some of the problems with fixed $q$ are directly reducible to solving a polynomial number of {\mwc} instances. However, this is not the case for {\ntiTj} and {\ntijTj}. Our $\alpha$-approximation algorithms for these are obtained by extending the CKR relaxation to a more general problem that we call {\lifted} (see Section~\ref{sec:lifted}) in such a way that the rounding methods used in \cite{SharmaVondrak} still give an $\alpha$-approximation. \lifted may have independent interest as a class of metric labeling problems that is broader than \mwc but can still be approximated to the same ratio. We then show that for fixed $q$, problems {\ntiTj} and {\ntijTj} are reducible to solving polynomially many instances of {\lifted}. 

\subparagraph*{Hardness of approximation.} We prove hardness of {\nstj} by reducing from \hitting. We then reduce {\ntijTj}, {\ntijtj}, and {\ntijtji} from {\nstj} to give hardness results for those problems as well.

%\vspace{-.3cm}
\subsection{Structure of the paper} In Section \ref{sec:background}, we present the main tools used in our algorithms and proofs. Section \ref{sec:lifted} introduces the \lifted problem and describes how to extend the $\alpha$-approximation of \cite{SharmaVondrak} to \lifted. The remaining sections present the results for the problems listed in Table \ref{tab}.

%\vspace{-.2cm}
%%%%%%%%%%%%%%%%
\section{Background}
\label{sec:background}
%%%%%%%%%%%%%%%%

Throughout the paper, we denote the set of non-negative reals by $\bR_+$, and use $[k]=\{1,\dots,k\}$. We use $e^i$ to denote $i$th elementary vector, and $\Delta_k$ denotes the convex hull of $\{e^1,\dots,e^k\}$, that is, $\Delta_k=\{x \in \mathbb{R}^k: x \geq 0,\ \sum_{i=1}^k x_i=1\}$.

Given an undirected graph $G=(V,E)$, the edge going between nodes $u,v\in V$ is denoted by $(u,v)$. For a weight function $w:E\to\bR_+$ and $C\subseteq E$, we use $w(C)=\sum_{e\in C} w(e)$. The graph obtained by deleting the edges in $C$ is denoted by $G-C$. We denote the set of components of $G$ by $\comp(G)$. The boundary of a given subset of nodes $S\subseteq V$ is $\delta(S)=\{(u,v)\in E: u\in S, v\in V\setminus S\}$. 
%The boundary of a partition $V_1,\dots, V_k$ of $V$ is $\cup_{j=1}^k \delta(V_j)$.

We briefly summarize the background results that we build upon in our proofs.

%%%%%%%%%%%%%%%%
\subsection{The CKR Relaxation and Rounding Methods}
%%%%%%%%%%%%%%%%

For a graph $G=(V,E)$ with edge weights $w:E\to\bR_+$ and terminals $S=\{s_1,\dots,s_k\}$, the CKR relaxation \cite{CKR} is the following linear program (CKR-LP) which assigns to each node $u \in V$ a geometric location $x^u$ in the $ k $-dimensional simplex.

\vspace{-1em}
\begin{alignat*}{2} 
    &\text{minimize} & & \sum_{(u,v)\in E} w_{u,v}\|x^u - x^v\|_1 \notag\\
    & \text{subject to}& \quad & \begin{aligned}[t]
        x^u & \in \Delta_k& u & \in V,\\
        x^{s_i} & = e^i& i & \in[k].
    \end{aligned} \tag{CKR-LP}
\end{alignat*}

  The original paper of C\v{a}linescu, Karloff and Rabani~\cite{CKR} gives a $ (3/2-1/k) $-approximation algorithm that works as follows. First take a threshold $ \rho_i\in (0,1) $ uniformly at random for each dimension $i\in[k]$. Then take one of the two permutations $ \sigma=(1,\ldots, k-1, k) $ and $ (k-1, k-2,\ldots, 1,k) $ of the terminals at random (that is, with probability $ 1/2 $), assign nodes within a distance $ \rho_{\sigma(i)} $ of $ x^{s_{\sigma(i)}} $ to the component of $ s_{\sigma(i)} $ for $ i\in[k-1]$, and assign the remaining nodes to $ s_k $. We call an algorithm that chooses a permutation of the terminals and then assigns the nodes within some threshold to the terminals in that order a \emph{threshold algorithm}. The analyses of the above linear programming formulation revealed several useful properties of the CKR relaxation. One of these observations is that the edges of the graph may be assumed to be \emph{axis-aligned}. An edge $u,v$ is said to be $(i,j)$-axis-aligned if $x^u$ and $x^v$ differ only in coordinates $i$ and $j$. Roughly speaking, any edge that is not axis-aligned can be subdivided into several edges that are axis-aligned, forming a piecewise linear path between $x^u$ and $x^v$. This observation significantly simplifies the analysis of threshold algorithms, as there are at most two thresholds that can cut any axis-aligned edge. Another useful property is \emph{symmetry}. For any threshold algorithm, there is one that achieves the same guarantees by choosing a uniformly random permutation. See \cite[Section 2]{Karger} for a more detailed discussion of these properties.

Another way of rounding the CKR relaxation is provided by the \emph{exponential clocks} algorithm of Buchbinder, Naor and Schwartz~\cite{buchbinder}. Their approach can be thought of 
as choosing a uniformly random point in the simplex, and splitting the simplex by axis parallel hyperplanes that meet at this given point. The algorithm gives the same guarantees as the algorithm of Kleinberg and Tardos~\cite{KleinbergTardos} for \uml. This latter problem takes as input a list of possible colors $\ell(v)$ for each node $v$ in a given graph, and asks for a coloring that respects these lists with the minimum total weight of dichromatic edges. Their relaxation (UML-LP) is similar to the CKR relaxation when there are a total of $q$ colors, but it does not require there to be nodes at every vertex of the simplex.

\vspace{-1em}
\begin{alignat*}{2}
    & \text{minimize} & & \sum_{(u,v)\in E} w_{u,v}\|x^u - x^v\|_1 \\
    & \text{subject to}& \quad & \begin{aligned}[t]
        x^u & \in \Delta_{q}& u & \in V,\\
        x^{v}_i & = 0& i & \notin \ell(v).
    \end{aligned} \tag{UML-LP}
\end{alignat*}

\begin{algorithm}[t!]
    \caption{The Kleinberg-Tardos Algorithm for \uml.}\label{alg:kleinberg}
    \begin{algorithmic}[1]
    \Statex \textbf{Input:} A graph $G=(V,E)$, weights $w:E\to\bR_+$, labels $\ell:V\to \mathcal{P}([q])$, and an LP solution $x^u$ for each $u\in V$.
    \Statex \textbf{Output:} A solution to \uml.
        \While{$\exists u\in V$ s.t. $u$ is unassigned}
        \State Pick a label $i\in [q]$ uniformly at random, and a threshold $\rho\sim unif[0,1]$.
        \State Assign label $i$ to any unassigned $u\in V$ with $x^u_i\geq \rho$.
        \EndWhile
    \end{algorithmic}
\end{algorithm}

It is shown in \cite[Section 6]{buchbinder} that Algorithm \ref{alg:kleinberg} gives the same guarantees as the exponential clocks algorithm.

The approximation algorithm of Sharma and Vondr\'ak~\cite{SharmaVondrak} for \mwc randomly chooses between four different algorithms of the above two types with some careful analysis to achieve an $ \alpha $-approximation, where $ \alpha \approx 1.2965 $.

%%%%%%%%%%%%%%%%
\subsection{Other Relevant Tools}
%%%%%%%%%%%%%%%%

Our hardness of approximation results are based on two different complexity assumptions. The $o(\log q)$ inapproximability results hold assuming $\mathrm{P} \neq \mathrm{NP}$, based on the hardness of approximating \hitting proved by Dinur and Steurer\cite{Dinur}. The other complexity assumption that we use is the \emph{Small Set Expansion Hypothesis} (SSEH), a core hypothesis for proving hardness of approximation for problems that do not have straightforward proofs assuming the Unique Games Conjecture (UGC). It implies the UGC, and we will use it as evidence against a $(2-\varepsilon)$-approximation, for any $\varepsilon>0$, for \kcut~\cite{kcut}. For completeness, we include the relevant theorems here.

\begin{theorem}[\cite{Dinur,moshkovitz}] For any fixed $0<\alpha<1$, 
    \hitting cannot be approximated in polynomial time within a factor of $(1-\alpha)\ln N$ on inputs of size $N$, unless $\mathrm{P}=\mathrm{NP}$.
\end{theorem}

\begin{theorem}[\cite{kcut}]
    Assuming the Small Set Expansion Hypothesis, it is NP-hard to approximate \kcut to within $(2-\varepsilon)$ factor of the optimum, for any constant $\varepsilon>0$.
\end{theorem}

From matroid theory, we use the notion of gammoids. A \emph{gammoid} $ M=(D,S,T) $ is a matroid defined by a digraph $ D=(V,E) $, a set of source nodes $ S\subseteq V $, and a set of target nodes $ T\subseteq V\setminus S $. A set $ X\subseteq T $ is independent in $ M $ if there exist $ |X| $ node-disjoint paths from elements of $ S $ into $ X $. 
Optimizing over a gammoid, as with any other matroid, can be done efficiently using the greedy algorithm.

Finally, Gomory-Hu (GH) tree \cite{gomory1961multi} is a standard tool in graph cut algorithms. The \emph{GH tree} of a graph $G=(V,E)$ with weight function $w:E\to\bR_+$ is a tree $T=(V,F)$ together with weight function $w_T:F\to\bR_+$ that  encodes the minimum-weight $s-t$ cuts for each pair $s,t$ of nodes in the following sense: the minimum $w_T$-weight of an edge on the $s-t$ path in $T$ is equal to the minimum $w$-weight of a cut in $G$ separating $s$ and $t$. Furthermore,  the two components of the tree obtained by removing the edge of minimum $w_T$-weight on the path give the two sides of a minimum $w$-weight $s-t$ cut in $G$.

%\vspace{-.2cm}
%%%%%%%%%%%%%%%%
\section{Lifted Cuts}
\label{sec:lifted}
%%%%%%%%%%%%%%%%

The goal of this section is to show that the following restriction of the \uml relaxation to a one-dimensional lifting of the CKR relaxation admits an $ \alpha $-approximation to its integer optimum. We define the lifted cut problem \lifted, which takes as input a graph $G=(V,E)$ with edge-weights $w:E\to\bR_+$, \emph{fixed} terminals $ S=\{s_1,s_2,\ldots,s_q\}\subseteq V $, and a list of possible colors for each node $ \ell:V\to \mathcal{P}[q+1] $, the power set of $[q+1]$, satisfying the following two conditions:
\begin{enumerate}[(A)]
    \item\label{item:1} $ \ell(s_i)=\{i\} $ for $ i=1\ldots q $,
    \item\label{item:2} $ q+1\in \ell(v) $ for all $ v\in V \setminus S $.
\end{enumerate} 
The goal is then to assign a color to each node from its list such that the total weight of dichromatic edges is minimized. We call the following linear programming relaxation of the \lifted problem LIFT-LP:
\begin{alignat*}{2}
    & \text{minimize} & & \sum_{(u,v)\in E} w_{u,v}\|x^u - x^v\|_1 \\
    & \text{subject to}& \quad & \begin{aligned}[t]
%			x^{t_i} & = e^i& i & \in[q],\\
        x^u & \in \Delta_{q+1},& u & \in V\\
        x^{u}_i & = 0& i & \notin \ell(u).
    \end{aligned}
\end{alignat*}
Condition \ref{item:1} ensures that the set $ S $ indeed defines terminals that vertices of the simplex are assigned to, as in \mwc, but Condition \ref{item:2} offers a relaxation, allowing a vertex of the simplex to not be assigned to any terminal. This condition gives an additional dimension to the simplex (see Figure \ref{fig:lift}), while still preserving the approximation guarantees given by the rounding algorithms for CKR.
%	This LIFT-LP differs from the CKR relaxation only in that it allows an additional 
\begin{figure}[ht]
       \centering
       \begin{tikzpicture}[scale=1.25]
           \tikzstyle{point}=[thick,inner sep=0.2pt,minimum width=4pt,minimum height=4pt]
           \tikzstyle{interior} = [thick,inner sep=0.2pt,minimum width=4pt,minimum height=4pt,circle]
           
           \node (a)[draw=red,point,label={[label distance=0cm]180:$s_1$}] at (-4,0) {};
           \node (b)[draw=green,point,label={[label distance=0cm]0:$s_2$}] at (-2,0) {};
           \node (d)[draw=blue,point,label={[label distance=0cm]0:$s_3$}] at (-3,0.7) {};
           \draw (a.center) -- (b.center);
           \draw (a.center) -- (d.center) -- (b.center);

           \node(1)[fill=green,interior] at (-2.2,.05) {};
           \node(2)[fill=blue,interior] at (-3.2,.53) {};
           \node(3)[fill=blue,interior] at (-2.8,.4) {};
           \node(1)[fill=red,interior] at (-3.5,.15) {};
           \node(1)[fill=black,interior] at (-2.65,.15) {};
           \node(1)[fill=black,interior] at (-3.4,.35) {};
           
           \draw[-{Latex[length=4mm,width=4mm]}] (-1.5,.7) -- (-0.5,.7);
           
           \node (a)[draw=red,point,label={[label distance=0cm]180:$s_1$}] at (0,0) {};
           \node (b)[draw=green,point,label={[label distance=0cm]0:$s_2$}] at (2,0) {};
           \node (c)[point,label={[label distance=0cm]0:}] at (1,2) {};
           \node (d)[draw=blue,point,label={[label distance=0cm]0:$s_3$}] at (1,0.7) {};
           \draw (a.center) -- (b.center) -- (c.center) -- cycle;
           \draw[dashed] (a.center) -- (d.center) -- (b.center);
           \draw[dashed] (d.center) -- (c.center);

           \node(1)[fill=green,interior] at (1.7,.6) {};
           \node(2)[fill=blue,interior] at (1,1.1) {};
           \node(3)[fill=blue,interior] at (1,1.5) {};
           \node(1)[fill=red,interior] at (.4,.8) {};
           \node(1)[fill=black,interior] at (1.65,.15) {};
           \node(1)[fill=black,interior] at (0.4,.35) {};
       \end{tikzpicture}
       \caption{An example of the original CKR relaxation in relation to our extended LIFT-LP on the case for {\tiTj}. The point colors represent the different candidate sets.}\label{fig:lift}
   \end{figure}
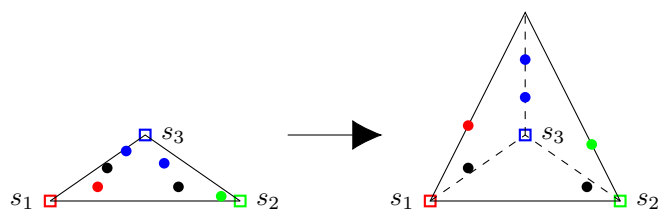

\begin{theorem}\label{thm:lift}
    The rounding scheme of \cite{SharmaVondrak}, when applied to LIFT-LP using Algorithm \ref{alg:kleinberg} in place of the exponential clocks algorithm, and with the modification that only the first $q$ coordinates are permuted in the threshold algorithms while coordinate $q+1$ is always left last, gives an $ \alpha $-approximation to \lifted.
\end{theorem}
\begin{proof}

First, we have to argue that the threshold algorithms give feasible solutions to \lifted (for Algorithm \ref{alg:kleinberg}, this follows since \lifted is a metric labeling problem). In all algorithms, $s_i$ is assigned to the $i$th component, since $x^{s_i}$ is the $i$th vertex of the simplex. For other nodes $v\in V \setminus S$, $x^{v}_i= 0$ guarantees that $v$ is not assigned to the $i$th component if $i\notin \ell(v)$. Here, we use the fact that the $(q+1)$st component  is the only one for which there is no threshold. Although it is possible that $x^{v}_{q+1}=0$ and $v$ is still assigned to the $(q+1)$st component, this is not a problem, because $ q+1\in \ell(v) $ by definition.

To prove that we have an $\alpha$-approximation, we need to show that 
%some variations of the lemmas on 
the relevant bounds that are used in the analysis of 
the four algorithms mentioned in Sharma-Vondr\' ak \cite{SharmaVondrak} 
    %	 Buchbinder et al. \cite{buchbinder} 
    carry through to this modified LIFT-LP. We give a sketch here, but the details are written out more carefully in Appendix \ref{app:ckr}. 
    We consider the two types of algorithms (i.e. threshold and exponential clocks) separately.

    It was observed in~\cite{buchbinder} that the exponential clocks algorithm can be replaced by the $2$-approximation for the \uml problem of Kleinberg-Tardos~\cite{KleinbergTardos}. Since LIFT-LP corresponds to a \uml problem, the bound in \cite[Lemma 3]{buchbinder} remains valid in our case. Since this is the relevant bound for the exponential clocks algorithm used in the analysis, we can conclude that Algorithm \ref{alg:kleinberg} for LIFT-LP gives the same guarantees.

    The other algorithms we need to consider are the threshold algorithms. These assume that there is a node at every vertex of the simplex, which is not necessarily true for the LIFT-LP as no variable needs to be at $ e^{q+1} $. We can however use that there is only \emph{one} such vertex, and change the order of the terminals so that this vertex is cut last.
    We can then use the analysis in \cite{SharmaVondrak} and \cite{buchbinder} of the CKR relaxation for $k=q+1$, as sketched below.

    The threshold algorithms first choose a random permutation of the nodes to achieve some symmetry, which is necessary for \emph{only} the first $ k-1 $ terminals. The last terminal, which is just assigned the remaining nodes, \emph{does not have its own threshold}. In each of the Single Threshold, Descending Threshold and Independent Threshold algorithms of \cite{SharmaVondrak}, thry prove results for the first two indices, and then argue that these hold for any pair of indices by symmetry. This is only directly clear for pairs in the first $ k-1 $. However, when 
 we consider an $(i,k)$-axis-aligned edge for some $i\in [k-1]$, the probability of cutting this edge can only be smaller as there is one less threshold to cut it; see \cite[Remark 2]{buchbinder} for a discussion. This reasoning holds even when there is no terminal at the $k$th vertex of the simplex.  
    
    Thus, the rounding scheme of \cite{SharmaVondrak}, with the modifications of using Algorithm \ref{alg:kleinberg} rather than exponential clocks and only permuting the first $q$ coordinates, gives an $ \alpha $-approximation for \lifted. For completeness, we include the relevant algorithms and lemmas from \cite{SharmaVondrak}, with the appropriate modifications, in Appendix \ref{app:ckr}.
    \end{proof}

%%%%%%%%%%%%%%%%
\section{{\ntiTj} Problem}
\label{sec:ti-Tj}
%%%%%%%%%%%%%%%%

In this problem, we are looking for a \emph{single} representative from each candidate set that will be separated from \emph{every} candidate in other candidate sets. This includes the other representatives, making the problem very similar to \mwc once the representatives are chosen. A key difference is that the optimal partition may have $ q+1 $ components.

We first look at the case where $ q $ is constant. 
    
    \begin{theorem}
        There is an $ \alpha $-approximation algorithm for \ntiTj when $ q $ is fixed.
    \end{theorem}    
    \begin{proof}
        First, guess the representative $t_i$ for each $i\in [q]$. As there are only $\prod_{i=1}^q |T_i| \leq n^q$ possible choices, this is polynomial in $n$ for fixed $q$. If a representative $t_i$ is in $T_j$ for some $j \neq i$, then there is obviously no solution. Otherwise, for a fixed choice of representatives, \ntiTj is an instance of \lifted. To see this, observe that the problem is equivalent to the \uml problem obtained by fixing the labels $ \ell(v) $ for $ v\in V $ as follows:
        \begin{enumerate}
            \item If $ v=t_i $ for some $ i $, then set $ \ell(v)\coloneqq \{i\} $.
            \item Otherwise, if $ v\in T_i\setminus\{t_i\} $ for a unique $i$, then set $ \ell(v)\coloneqq\{i,q+1\} $.
            \item If $ v\in T_i \cap  T_j\setminus\{t_i,t_j\}$ for $i \neq j$, then set $ \ell(v)\coloneqq\{q+1\} $.
            \item Finally, if $ v\in V\setminus \bigcup_{i\in[q]} T_i $, then set $ \ell(v) \coloneqq [q+1] $.
        \end{enumerate} 
        This is an instance of \lifted : Condition \ref{item:1} is a clear consequence of the first rule, and since any node that is not a representative has $ q+1 $ as one of its labels, Condition \ref{item:2} follows as well.
        Therefore, Theorem~\ref{thm:lift} leads to an $ \alpha $-approximation.%\hfill\qed
\end{proof}

Following the idea of the classical $ 2 $-approximation for \mwc discussed in the introduction, there is a simple $ 2 $-approximation when $ q $ is arbitrary. 

\begin{theorem}
    There is a $ 2 $-approximation algorithm for {\ntiTj}.
\end{theorem}
\begin{proof}
    For each candidate set $ T_i $, let $ t_i\in T_i$ be a node for which the minimum-weight cut separating $t_i$ from  $\cup_{j\neq i} T_j $ is as small as possible, and let $C$ be the union of these isolating cuts. To see that the solution is within a factor $ 2 $ of the optimum, consider an optimal solution to \ntiTj and let $V_1,V_2,\dots, V_q,V_{q+1}$ denote the components after its deletion, where $ V_{q+1} $ may be empty and the components are ordered by the index of the representative they contain. The boundary of each $ V_i $ is an isolating cut of some candidate in $ T_i $, which the algorithm minimized. Summing up the weights of the boundaries, we count each edge twice, and the theorem follows.%\hfill\qed
    \end{proof}

%%%%%%%%%%%%%%%%
\section{{\ntitj} Problem}
\label{sec:titj}
%%%%%%%%%%%%%%%%

In this problem, we are looking for a \emph{single} representative from each candidate set together with a minimum multiway cut separating them. Note that when $ T_1=T_2=\ldots=T_q=V $, \ntitj generalizes \kcut where we seek the minimum-weight cut that partitions the graph into $ k $ parts. It is known that \kcut is hard to approximate within a factor of $2-\varepsilon$ for any $\varepsilon>0$, assuming SSEH \cite{kcut}. 

\begin{theorem}
There is an $\alpha$-approximation for \ntitj when $q$ is fixed.    
\end{theorem}
\begin{proof}
When $ q $ is fixed, one can iterate through all the $O(n^q)$ possible choices of representatives, approximate the corresponding \mwc instance, and choose the best one. %\hfill\qed     
\end{proof}

For general $q$, it is helpful to first look at the case where $G$ is a tree. We show that in this special case, the problem reduces to finding the minimum cost basis of a gammoid. We call a cut $C\subseteq E$ \emph{good} if $G-C$ has a valid set of representatives, that is, if we can choose $|C|+1$ representatives that form a partial transversal of the candidate sets, and each component of $G-C$ contains a single representative from this partial transversal. The algorithm is presented as Algorithm~\ref{alg:tree_greedy}.

\begin{algorithm}[t!]
    \caption{Greedy algorithm for \ntitj on trees.}\label{alg:tree_greedy}
    \begin{algorithmic}[1]
    \Statex \textbf{Input:} A tree $G=(V,E)$, weights $w:E\to\bR_+$, candidates $T_1,\dots,T_q\subseteq V$.
    \Statex \textbf{Output:} A minimum-weight good cut $C\subseteq E$. 
        \State Set $C \gets \emptyset$.
        \While{$|C|< q-1$}
        \State $e\gets \argmin\{ w(e) : e\notin C, C+e\text{ is \emph{good}} \}$
        \State $C = C+e$
        \EndWhile
    \end{algorithmic}
\end{algorithm}

\begin{theorem}
    Algorithm~\ref{alg:tree_greedy} computes an optimal solution to {\ntitj} on trees.
\end{theorem}
\begin{proof}
    We prove the statement by showing that the problem is equivalent to optimizing over a gammoid. We construct a directed graph as follows. Let $r\in V$ be an arbitrary root node, and orient the edges of the tree towards $r$.
    %, thus resulting in an arborescence.
    For a non-root node $v$, we denote the unique arc leaving $v$ by $e(v)$ and define the cost of $v$ to be $w(e(v))$. Furthermore, for each set $T_i$, we add a node $s_i$ together with arcs from $s_i$ to the candidates in $T_i$. 
    
    Let $D$ denote the digraph thus obtained, $S\coloneqq\{s_1,\dots,s_q\}$, and $T\coloneqq V$, and consider the gammoid $M=(D,S,T)$. The key observation is the following.

    \begin{claim}\label{cl:2}
        For a set $Z\subseteq V\setminus\{r\}$, $C=\{e(v): v\in Z\}$ is a good cut if and only if $Z\cup\{r\}$ is independent in $M$.
    \end{claim}
    \begin{proof}
        For the forward direction, assume that $C=\{e(v): v\in Z\}$ forms a good cut. Let $Z=\{v_1,\dots,v_p\}$. Without loss of generality, we may assume that the candidate sets having a valid set of representatives in $G-C$ are $T_1,\dots,T_p,T_{p+1}$, where $v_i$ is in the component of the representative $t_i$ of $T_i$ and $r$ is in the same component as the representative $t_{p+1}$ of $T_{p+1}$. For $i\in[p]$, the edge $(s_i,t_i)$ and the path $t_i$-$v_i$ in the tree form an $s_i$-$v_i$ path; similarly, the edge $(s_{p+1},t_{p+1})$ and the path $t_{p+1}$-$r$ in the tree form an $s_{p+1}$-$r$ path. Furthermore, these paths are pairwise node-disjoint, since they use different connected components of $G-C$.
        
        For the other direction, assume that $Z\cup\{r\}$ is independent in $M$, and let $Z=\{v_1,\dots,v_p\}$. Without loss of generality, we may assume that there are pairwise node-disjoint paths from $s_i$ to $v_i$ for $i\in[p]$ together with a path from $s_{p+1}$ to $r$. Let $t_i\in T_i$ be the first node on the path starting from $s_i$ for $i\in[p+1]$. Then $\{t_1,\dots,t_{p+1}\}$ form a valid system of distinct representatives for the cut $C$ as  
        each of these nodes are in a separate component of $G-C$.%\hfill\qed
        \end{proof}
    
    By Claim \ref{cl:2}, a minimum-weight good cut can be determined using the greedy algorithm for matroids, which is exactly what Algorithm~\ref{alg:tree_greedy} is doing.%\hfill\qed
    \end{proof}

Algorithm~\ref{alg:tree_greedy} solves the special case when $ G $ is a tree. The classical $ (2-2/k) $ approximation for \mwc \cite{Dahlhausetal} uses 2-way cuts coming from the Gomory-Hu tree, and so does the $ (2-2/k) $ approximation for \kcut \cite{kcut2apx}. We follow a similar approach in Algorithm~\ref{alg:gh_greedy}. The algorithm can be interpreted as taking the minimum edges in the GH tree as long as they allow a valid system of representatives. The algorithm is presented as Algorithm~\ref{alg:gh_greedy}.

\begin{algorithm}[t!]
    \caption{Approximation algorithm for \ntitj on graphs.}\label{alg:gh_greedy}
    \begin{algorithmic}[1]
    \Statex \textbf{Input:} A graph $G=(V,E)$, weights $w:E\to\bR_+$, candidates $T_1,\dots,T_q\subseteq V$.
    \Statex \textbf{Output:} A feasible cut $C\subseteq E$. 
        \State Compute the Gomory-Hu tree $ H $ of $ G $.
        \State Run Algorithm \ref{alg:tree_greedy} on $ H $.
        \State Return the union of the cuts corresponding to edges found in Step 2.
    \end{algorithmic}
\end{algorithm}

%\vspace{-.1cm}
\begin{theorem}\label{thm:ti-tj}
    Algorithm~\ref{alg:gh_greedy} computes a $ (2-2/q) $ approximation to {\ntitj} on arbitrary graphs.
\end{theorem}
\begin{proof}
    Let $ OPT $ be the optimal solution with representatives $ t_1^*,\ldots, t_q^* $, and components $ V_1^*,\ldots,V_q^* $, where $ V_q^* $ has the maximum weight boundary $ \delta(V_q^*) $. Let also $ H $ be the GH tree of $ G $.

    We transform $OPT $ into a solution $OPT_{GH}$ on $H$, losing at most a factor of $(2-2/q)$.
    We do this by repeatedly removing the minimum weight edge in $E(H)$ that separates a pair among the representatives $t_1^*,\ldots, t_q^*$ that are in the same component of $H$. More precisely, we start with $H_0=H$, and take the minimum-weight edge $ e_1\in E(H_0) $ separating some pair of representatives $t_{i}^*, t_{j}^*$ in $OPT$ that are in the same component of $ H_0 $. Define the edge $f_1=(t_{i}^*, t_{j}^*) $. Then we construct $ H_1=H_0-e_1 $, and repeat this process to get a sequence of edges $ e_1,e_2,\ldots, e_{q-1} $ and a tree of representative pairs $ F=(\{ t_1^*,\ldots, t_q^* \},\{f_1,\ldots, f_{q-1}\}) $.

    Direct the edges of $ F $ away from $ t_q^* $, and reorder the edges such that $ f^1 $ is the edge going into $ t_1^* $, $ f^2 $ into $ t_2^* $, and so on. Let $ e^i $ be the edge of the GH tree corresponding to $ f^i $, i.e., the minimum weight edge of the path between the two endpoints of $ f^i $, and let $ U(e^i) $ be the cut corresponding to $ e^i $ for each $ i $.
    Then the boundary of each component satisfies $ w(\delta(V_i^*)) \geq w(U(e^i)) $, as $ \delta(V_i^*) $ separates the two representatives in $ f^i $ as well, and $ U(e^i) $ is the minimum-weight cut between these. 
    
    Let the solution $ OPT_{GH} $ be $ \bigcup_{i\in[q-1]} U(e_i) $,  $ ALG $ the cut found by the algorithm, $ ALG_{GH} $ the corresponding edges in the GH tree $ H $, and $ w_H $ the weight function on $ H $. Then
    \begin{align*}
        w(ALG)&\leq w_{H}(ALG_{GH}) \leq w_H(OPT_{GH}) =  \sum_{i=1}^{q-1} w(U(e^i)) \leq \sum_{i=1}^{q-1} w(\delta(V^*_i)) \\
        &\leq (1-1/q) \sum_{i=1}^{q} w(\delta(V^*_i)) \leq (2-2/q)w(OPT).
    \end{align*}
% \hfill\qed
\end{proof}

%%%%%%%%%%%%%%%%
\section{{\nstj}, {\ntijtj}, {\ntijtji}, and \ntijTj Problems}
\label{sec:stj}
%%%%%%%%%%%%%%%%

In this section, we combine the study of four problems, as the techniques are similar. 
%\vspace{-.4cm   }

\subsection{Hardness of approximation}

All four have similar proofs of hardness of approximation, which we state here but leave the proofs to appendix \ref{app:hard} for brevity. 

\begin{theorem}\label{thm:hard}
For general $q$, \nstj, \ntijtj and \ntijTj are at least as hard to approximate as \hitting.
\end{theorem}

We omit the \ntijtji problem from Theorem \ref{thm:hard} because it follows as a corollary to Theorem \ref{thm:equiv}, which states that it is equivalent to the known \textsc{Steiner Multicut} problem. The conditional $o(\log n)$ inapproximability was already proved for \textsc{Steiner Multicut} in \cite{reqhardness}, using similar instances as those in our proof of Theorem \ref{thm:hard}.
The \ntijtji problem asks to find a cut such that each pair of candidate sets have at least one element in separate components, where this choice can depend on the pair. 

\begin{theorem}\label{thm:equiv}
    The \ntijtji problem is equivalent to \textsc{Steiner Multicut}.
\end{theorem}
\begin{proof}
    To reduce from \textsc{Steiner Multicut}, we are given $q$ subsets $X_0,X_1,\ldots, X_{q-1}$ of nodes of a graph $G$, each of which needs to be cut into at least two components. 
    We construct a \ntijtji instance on the same graph with $2q$ candidate sets $T_0,T_1,\ldots,T_{2q-1}$, where $T_i = X_{\lfloor i/2 \rfloor}$ for $0\leq i \leq 2q-1$. Then, for each $j=0\ldots q-1$, the condition that $T_{2j}$ must be separated from $T_{2j+1}$ ensures that there are two nodes $t_{2j}^{2j+1}, t_{2j+1}^{2j} \in X_j$ that are in different components. 
    In other words, the solution is a minimal cut that, once removed, divides each set into at least two components. 
    If the conditions of \ntijtji hold for $T_{2j}$ and $T_{2j+1}$ for any $j$, then they hold automatically for any other pair of candidate sets too,  
    because once a set has elements in two components, at least one of them will be in a different component than some element of any given candidate  set.

    For the other direction, we are given $q$ subsets $T_1,\dots, T_q$ of nodes of a graph $G$ as a \ntijtji instance. We then make a \textsc{Steiner Multicut} instance with $\binom{q}{2}$ vertex sets indexed by pairs $i, j\in [q]^2$ with $i\neq j$.
    The set $X_{i,j}$ will then be $T_i\cup T_j$, which means any valid \textsc{Steiner Multicut} solution $C$ will split each of these sets into at least two components. We claim $C$ is a valid \ntijtji solution as well. Let $v_{i,j},u_{i,j}\in X_{i,j}$ be in different components of $G\setminus C$.
    Then one of the following cases must 
    hold:
    \begin{enumerate}
        \item $v_{i,j}\in T_i$ and $u_{i,j}\in T_j$. In this case, let $t_j^i \coloneqq u_{i,j}$ and $t_i^j\coloneqq v_{i,j}$.
        \item $u_{i,j}\in T_i$ and $v_{i,j}\in T_j$. In this case, let $t_j^i \coloneqq v_{i,j}$ and $t_i^j\coloneqq u_{i,j}$.
        \item $u_{i,j},v_{i,j}\in T_i$. Then either 
        \begin{enumerate}[(i)]
            \item all of $T_j$ is in the same component of $G\setminus C$ as $u_{i,j}$, in which case let $t_i^j\coloneqq v_{i,j}$, and set $t_j^i$ to an arbitrary element of $T_j$, or
            \item some vertex $w\in T_j$ is in a different component of $G\setminus C$ than $u_{i,j}$, in which case let $t_j^i\coloneqq w$, and $t_i^j\coloneqq u_{i,j}$.
        \end{enumerate}
        \item Similarly, if $u_{i,j},v_{i,j}\in T_j$, then either 
        \begin{enumerate}[(i)]
            \item all of $T_i$ is in the same component of $G\setminus C$ as $u_{i,j}$, in which case let $t_j^i\coloneqq v_{i,j}$, and set $t_i^j$ to an arbitrary element of $T_i$, or
            \item some vertex $w\in T_i$ is in a different component of $G\setminus C$ than $u_{i,j}$, in which case let $t_i^j\coloneqq w$, and $t_j^i\coloneqq u_{i,j}$.
        \end{enumerate}
    \end{enumerate}
    In all cases above, $t_j^i$ is in a different component than $t_i^j$ on $G\setminus C$, so $C$ is a valid \ntijtji solution. Any \ntijtji solution is clearly also a solution for this \textsc{Steiner Multicut} instance, so the optimal cut is the same for both.%\hfill\qed
\end{proof}

\subsection{Fixed $q$}

The techniques when $q$ is fixed differ, suggesting that the problems themselves are quite different, despite the apparent similarities. 

\paragraph*{Fixed Terminal}
The \nstj problem is slightly different from the others, as the goal here is to choose representatives that need to be separated only from a \emph{fixed} node $ s $. In this case, the problem can be solved efficiently.

\begin{proposition}\label{prop:poly}
For fixed $q$, \nstj can be solved in polynomial time.
\end{proposition}

\begin{proof}
    In this case, one can iterate through all possible choices of representatives, of which we have at most $n^q$, calculate a minimum two-way $ s-\{t_i: i\in[q]\}$ cut for each, and then take the best of all solutions.
\end{proof}

\paragraph*{Some to single/some} The {\ntijtj} and {\ntijtji} problems both become \mc instances with a constant number of terminals in this case, which gives the following theorem:

\begin{theorem}\label{thm:ntijtjifix}
For fixed $q$, there is an $\alpha$-approximation to {\ntijtj} and {\ntijtji}.
\end{theorem}
\begin{proof}
    We will use the $\alpha$-approximation to \mwc on a polynomial number of instances with fixed terminals.
    We begin with the {\ntijtj} problem. In this problem, the goal is to choose a single representative $t_j$ for each $j\in[q]$ together with some candidate $t_i^j\in T_i$ for each pair $i\neq j$ that are then separated by the cut.
    
    When $ q $ is fixed, one can guess the representatives $ t_i^j $ and $ t_j $ to get a set of terminals $ S $ together with some separation demands on them. The number of such terminals can be bounded as $ |S|\leq q^2 $.
%, but it need not be as large. 
A slightly more careful analysis shows that the number of different $t_i^j$ nodes for a candidate set $T_i$ can be bounded by two. Thus, we only have to guess three representatives from each $ T_i $, implying $ |S|\leq 3q $. Either way, the number of guesses for $ S $ is polynomial in $n$. Each guess of $S$ defines a minimum multicut problem since we know which pairs of representatives have to be separated. We can compute an $\alpha$-approximation to each of these \mc problems by enumerating all possible partitions of $S$ (of which there are exponentially many in $q$) that satisfy the multicut demands, collapsing the partitions into fixed terminals, and calculating an $\alpha$-approximating multiway cut for each.

For the {\ntijtji} problem, again guess the representatives $t_i^j$ for each $i,j\in[q],i\neq j$ to get a set of terminals $S$ together with some separation demands on them. Since any candidate set with terminals in different components already has at least one element in a separate component for any other candidate set, the number of such terminals can be bounded by $|S|\leq 2q$. For each fixed $S$, we can find an $\alpha$-approximation the same way as above.
\end{proof}

Combining this approximation for \ntijtji with Theorem \ref{thm:equiv} gives the current best approximation for \textsc{Steiner Multicut} in the regime where the number of candidates depends on $n$, and the number of sets is constant.

\paragraph*{Some to all}
Finally, we consider the \ntijTj problem, which asks to find representatives $t_i^j\in T_i$ for each pair $ i,j\in[q] $ and a minimum-weight cut $C\subseteq E$ such that $ t_i^j $ is separated from \emph{all} of $ T_j $ in $G-C$. The case for constant $q$ uses the tool from Section \ref{sec:lifted}.

\begin{theorem}
    There is an $ \alpha $-approximation for {\ntijTj} when $ q $ is fixed.
\end{theorem}
\begin{proof}
    We guess all representatives $ t_i^j $; there are at most $ n^{q^2} $ possible choices, which is polynomial if $q$ is fixed. Note that we may assume that $t_i^j \neq t_j^{\ell}$ if $i \neq j$, otherwise there is obviously no solution. We also guess a valid partition $ V_1,\ldots, V_{q_1} $ of these representatives into $ q_1 $ components, where $ 2\leq q_1 \leq q^2 $ (validity means that $t_i^j$ and $t_j^{\ell}$ are in different classes of the partition if $i\neq j$). The number of such partitions is exponential in $ q $, but we can still enumerate them when $ q $ is fixed (note that this is not a partition of $V$, but a partition of the set of all chosen representatives, which is a vertex set of size at most $q^2$).  
For such a partition, the problem becomes an instance of $(q_1+1)$-dimensional \lifted with the following labels.
\begin{enumerate}[a)]
    \item If $ v\in V_k $ for some $ k \in [q_1] $, then set $ \ell(v)\coloneqq \{k\} $. \label{item:a}
    \item Otherwise, if $ v\in T_j $, then we must ensure that the label cannot be any partition containing some $ t_i^j $. In other words, set $ \ell(v)\coloneqq \{1,2,\ldots,q_1+1\}\setminus \{k: v \in T_j  \text{ and } t_i^j\in V_k \text{ for some }i,j\} $. \label{item:b}
    \item Finally, if $ v\in V\setminus \bigcup_{i\in[q]} T_i $, then set $ \ell(v) \coloneqq [q_1+1]$. \label{item:c}
\end{enumerate}
Conditions \ref{item:1} and \ref{item:2} of \lifted are not difficult to verify. The solution to this problem is a solution to {\ntijTj}. Indeed, consider the partition given by a solution to \lifted , which is an extension of the partition $ V_1,\ldots, V_{q_1} $ by condition \ref{item:a}, with an additional class for label $q_1+1$.
Condition \ref{item:b} then ensures, for a given $ t_i^j $, that the component of $ t_i^j $ cannot contain any element of $ T_j $. Thus, Theorem \ref{thm:lift} gives an $ \alpha $-approximation for {\ntijTj} as well.
\end{proof}

\bibliography{MFCS/multiway_MFCS_final}
\newpage

\appendix

\section{Hardness Proofs for \nstj, \ntijtj, and \ntijTj}\label{app:hard}

Here we prove Theorem \ref{thm:hard}. In the \hitting problem, the input is a family $ S_1,\ldots,S_m $ of sets, and the goal is to find a smallest set of elements that intersects each of them. 

\begin{proposition}\label{thm:stjhard}
For general $q$, \nstj is at least as hard to approximate as \hitting.
\end{proposition}
\begin{proof}
To reduce \hitting to \nstj, we create a graph $ G $ by adding a node $s$ together with edges of weight $1$ from $s$ to the elements of the ground set. Thus we get a star with center $ s $, and choose the candidate sets to be the sets $ S_1,\ldots,S_m $. Then, minimizing the number of edges needed to separate at least one node of each $S_i$ from $s$ is equivalent to finding a minimum hitting set; see Figure \ref{fig:hit} for an illustration. Note that the reduction is approximation factor preserving. Since there is no $ o(\log m) $ approximation for \hitting assuming $\mathrm{P} \neq \mathrm{NP}$ \cite{Dinur,moshkovitz}, the hardness of \nstj follows. %\hfill\qed
\end{proof}

\begin{figure}[t!]
\begin{minipage}{.48\textwidth}
   \centering
   \vspace{.185in}
   % \resizebox{1.0\textwidth}{!}{%
       \begin{tikzpicture}
           \node[external] (r)  at (0,0) {$ s $};
           \node[internal] (1) [below left=1cm and 2cm of r] {};
           \node[internal] (2) [below left=1cm and 1.2cm of r] {};
           \node[internal] (3) [below left=1cm and .4cm of r] {};
           \node[internal] (4) [below=.83cm of r] {};
           \node[internal] (5) [below right=1cm and .4cm of r] {};
           \node[internal] (6) [below right=1cm and 1.2cm of r] {};
           %							\node[internal, fill=white!0, draw=none] (dots) [below left=1cm and -1cm of r] {\ldots};
           \node[internal] (7) [below right=1cm and 2cm of r] {};
           \draw[-] (r) -- (1);
           \draw[-] (r) -- (2);
           \draw[-] (r) -- (3);
           \draw[-] (r) -- (4);

           \draw[-] (r) -- (5);
           \draw[-] (r) -- (6);
           \draw[-] (r) -- (7);
           
           \begin{scope}[fill opacity=0.4]
               \filldraw[fill=yellow!70] ($(1)+(-0.5,0)$) 
               to[out=90,in=90,looseness=.5] ($(3) + (0.5,0)$)
               to[out=270,in=270,looseness=.5] ($(1) + (-0.5,0)$);
               \filldraw[fill=red!70] ($(3)+(-0.5,0)$) 
               to[out=90,in=90,looseness=.75] ($(4) + (0.5,0)$)
               to[out=270,in=270,looseness=.75] ($(3) + (-0.5,0)$);
               \filldraw[fill=green!70] ($(3)+(-0.5,0)$) 
               to[out=90,in=90,looseness=.5] ($(6) + (0.5,0)$)
               to[out=270,in=270,looseness=.5] ($(3) + (-0.5,0)$);
               \filldraw[fill=blue!70] ($(5)+(-0.5,0)$) 
               to[out=90,in=90,looseness=.5] ($(7) + (0.5,0)$)
               to[out=270,in=270,looseness=.5] ($(5) + (-0.5,0)$);
           \end{scope}
       \end{tikzpicture}
       \caption{A picture of the reduction for \nstj from \hitting. 
   % The candidate sets are shown in colorful ovals, and the minimum {\stj} cut will be the minimum hitting set.
   }\label{fig:hit}
       % }
   \end{minipage}%
   \hfill
   \begin{minipage}{.48\textwidth}
   \centering
       % \resizebox{1.0\textwidth}{!}{%
           \begin{tikzpicture}
               \node[external] (r)  at (0,0) {$ s $};
               \node[internal] (1) [below left=1cm and 2cm of r] {};
               \node[internal] (2) [below left=1cm and 1.2cm of r] {};
               \node[internal] (3) [below left=1cm and .4cm of r] {};
               \node[internal] (4) [below=.83cm of r] {};
               \node[internal] (5) [below right=1cm and .4cm of r] {};
               \node[internal] (6) [below right=1cm and 1.2cm of r] {};
               %							\node[internal, fill=white!0, draw=none] (dots) [below left=1cm and -1cm of r] {\ldots};
               \node[internal] (7) [below right=1cm and 2cm of r] {};
               \draw[-] (r) -- (1);
               \draw[-] (r) -- (2);
               \draw[-] (r) -- (3);
               \draw[-] (r) -- (4);
               \draw[-] (r) -- (5);
               \draw[-] (r) -- (6);
               \draw[-] (r) -- (7);
               
               \begin{scope}[fill opacity=0.3]
                   \filldraw[fill=yellow!70] ($(1)+(-0.5,0)$) 
                   to[out=75,in=225,looseness=.5] ($ (r) + (-.5,.5) $)
                   to[out=45,in=45, looseness=1.707] ($ (r) + (.5,-.5) $)
                   to[out=225,in=45,looseness=.5] ($(3) + (0.5,0)$)
                   to[out=270,in=270,looseness=.5] ($(1) + (-0.5,0)$);
                   \filldraw[fill=red!70] ($(3)+(-0.5,0)$) 
                   to[out=90,in=225,looseness=.5] ($ (r) + (-.707,0) $)
                   to[out=90,in=90, looseness=1.707] ($ (r) + (.707,0) $)
                   to[out=270,in=90,looseness=.5] ($(4) + (0.5,0)$)
                   to[out=270,in=270,looseness=.75] ($(3) + (-0.5,0)$);
                   \filldraw[fill=green!70] ($(3)+(-0.5,0)$) 
                   to[out=90,in=225,looseness=.5] ($ (r) + (-.707,0) $)
                   to[out=90,in=90, looseness=1.707] ($ (r) + (.707,0) $)
                   to[out=315,in=90,looseness=.5] ($(6) + (0.5,0)$)
                   to[out=270,in=270,looseness=.5] ($(3) + (-0.5,0)$);
                   \filldraw[fill=blue!70] ($(5)+(-0.5,0)$) 
                   to[out=135,in=315,looseness=.5] ($ (r) + (-.5,-.5) $)
                   to[out=135,in=135, looseness=1.707] ($ (r) + (.5,.5) $)
                   to[out=315,in=135,looseness=.5] ($(7) + (0.5,0)$)
                   to[out=270,in=270,looseness=.5] ($(5) + (-0.5,0)$);
                   %									\filldraw[fill=brown] ($ (r) + (.5,0) $)
                   %									to[out=90,in=90, looseness=1] ($ (r) + (-.5,0) $)
                   %									to[out=180,in=180, looseness=1] ($ (r) + (.5,0) $);
               \end{scope}
               
               \node[external, fill=red!30] (redraw)  at (0,0) {$ s $};
           \end{tikzpicture}
   \caption{A picture of the reduction from {\nstj} to {\ntijtj}.}\label{fig:red}
   % }
   \end{minipage}%
\end{figure}

\begin{proposition}\label{thm:tij-tj}
    For general $ q $, {\ntijtj} is at least as hard to approximate as {\nstj}.
\end{proposition}
\begin{proof}
    Given an instance $s,T_1,\ldots,\allowbreak T_q $ of {\nstj}, we create an instance of {\ntijtj} with $ T_i'\coloneqq T_i\cup \{s\} $ for $ i\in[q] $, and $ T_{q+1}'\coloneqq \{s\} $; see Figure \ref{fig:red} for an example.

    Given a solution to the {\nstj} instance, we can obtain a solution of the same weight to the {\ntijtj} instance by keeping the representatives $t_i$ for $i\in [q]$, setting $t_{q+1}=t^i_{q+1} \coloneqq s$ for $i\in [q]$, $t^{q+1}_i\coloneqq t_i$ for $i\in [q]$, and $t_i^j\coloneqq s$ for $ i,j\in [q],i\neq j $.

    For the other direction, we observe that 
    each $t_j$ ($j\in [q]$) must be separated from $s$ in a solution of the {\ntijtj} instance. Thus, we obtain a solution with the same weight for the {\nstj} if we keep the same representatives $t_j$ ($j\in [q]$).
    \end{proof}

\begin{proposition}\label{thm:tij-Tj}
    For general $ q $, {\ntijTj} is at least as hard to approximate as {\nstj}.
\end{proposition}
\begin{proof}
    Given an instance of {\nstj} with sets $ T_1,\ldots,T_q $ on a graph $ G=(V,E) $ where $q \geq 2$, we construct a {\ntijTj} instance as follows. We add additional nodes $ V_0=\{s_1,s_2,\ldots, s_q\} $, $ G'=(V\cup V_0, E) $, $ T_i'=T_i\cup \{s_i\} $ for $ i=1\ldots q $, and $ T_{q+1}' = \{s,s_1,s_2,\ldots, s_q\} $; see Figure \ref{fig:tij-Tj} for an example.
    
   \begin{figure}[ht]
       \centering
       \resizebox{!}{.2\textwidth}{%
           \begin{tikzpicture}
               \node[external] (r)  at (0,0) {$ s $};
               \node[internal] (1) [below left=1cm and 2cm of r] {};
               \node[internal] (2) [below left=1cm and 1.2cm of r] {};
               \node[internal] (3) [below left=1cm and .4cm of r] {};
               \node[internal] (4) [below=.83cm of r] {};
               \node[internal] (5) [below right=1cm and .4cm of r] {};
               \node[internal] (6) [below right=1cm and 1.2cm of r] {};
               
               \node[external, fill opacity=0.4, fill=yellow!70] (s1) [below=0.5cm of 2] {$ \scriptsize s_1 $};
               \node[external, fill opacity=0.4, fill=red!70] (s2) [right=0.4cm of s1] {$ \scriptsize s_2 $};
               \node[external, fill opacity=0.4, fill=green!70] (s2) [right=0.4cm of s1] {$ \scriptsize s_2 $};
               \node[external, fill opacity=0.4, fill=green!70] (s3) [right=0.05cm of s2] {$ \scriptsize s_3 $};
               \node[external, fill opacity=0.4, fill=blue!70] (s4) [below=0.5cm of 6] {$ \scriptsize s_4 $};
               %							\node[internal, fill=white!0, draw=none] (dots) [below left=1cm and -1cm of r] {\ldots};
               \node[internal] (7) [below right=1cm and 2cm of r] {};
               \draw[-] (r) -- (1);
               \draw[-] (r) -- (2);
               \draw[-] (r) -- (3);
               \draw[-] (r) -- (4);
               \draw[-] (r) -- (5);
               \draw[-] (r) -- (6);
               \draw[-] (r) -- (7);
               
               \begin{scope}[fill opacity=0.4]
                   \filldraw[fill=yellow!70] ($(1)+(-0.5,0)$) 
                   to[out=90,in=90,looseness=.5] ($(3) + (0.5,0)$)
                   to[out=270,in=270,looseness=.5] ($(1) + (-0.5,0)$);
                   \filldraw[fill=red!70] ($(3)+(-0.5,0)$) 
                   to[out=90,in=90,looseness=.75] ($(4) + (0.5,0)$)
                   to[out=270,in=270,looseness=.75] ($(3) + (-0.5,0)$);
                   \filldraw[fill=green!70] ($(3)+(-0.5,0)$) 
                   to[out=90,in=90,looseness=.5] ($(6) + (0.5,0)$)
                   to[out=270,in=270,looseness=.5] ($(3) + (-0.5,0)$);
                   \filldraw[fill=blue!70] ($(5)+(-0.5,0)$) 
                   to[out=90,in=90,looseness=.5] ($(7) + (0.5,0)$)
                   to[out=270,in=270,looseness=.5] ($(5) + (-0.5,0)$);
                   
                   \filldraw[fill=red!70] ($ (r)+(0,-.5) $) to[out=180,in=180,looseness=2]
                   ($ (r)+(0,.5) $)
                   to[out=0,in=90,looseness=.5]
                   ($ (7)+(1,0) $)
                   to[out=270,in=0,looseness=.5]
                   ($ (s4)+(0,-.5) $)
                   to[out=180,in=270,looseness=0]
                   ($ (s1)+(0,-.5) $) to[out=180,in=180,looseness=2]
                   ($ (s1)+(0,.5) $)
                   to[out=0,in=180,looseness=0]
                   ($ (s4)+(0,.5) $)
                   to[out=0,in=270,looseness=.5]
                   ($ (7)+(.6,0) $)
                   to[out=90,in=0,looseness=.5]
                   ($ (r)+(0,-.5) $);
               \end{scope}

           \end{tikzpicture}
       }
       \caption{A picture of the reduction from {\nstj}. The candidate set $ T_1 $ is in yellow, $ T_2 $ in brown, $ T_3 $ in green, $ T_4 $ in blue, and $ T_5 $ in red. }\label{fig:tij-Tj}
   \end{figure}
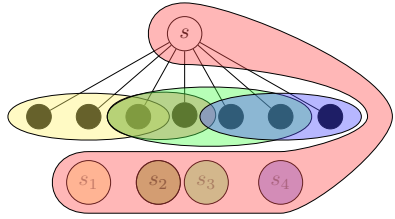

    Given a {\nstj} solution with representatives $ t_1^*,\ldots, t_{q}^* $, we get a solution to this instance as follows:
$t_{q+1}^j=s_{(j+1)\mod q}$ for $j \in [q+1]$; if $i\in [q]$, then $ t_i^{q+1}=t_i^*$, and $ t_i^j=s_i$ for $j\in [q]$. 
Then the same cut will separate each $ t_i^j $ from all of $ T'_j $, and have the same weight.
    
    Given an optimal solution to the {\ntijTj} instance, we can assume without loss of generality that $ t_i^j= s_{(j+1)\mod q} $ when $ i={q+1}, $ and  $ t_i^j=s_i $ when $ i\neq q+1, j\neq q+1 $, as these are separated from the corresponding $ T'_j $ in $ G' $. Then we can get a {\nstj} solution by setting $ t_j=t_j^{q+1} $ for all $j\in [q]$ and removing the same edges. This reduction preserves approximation, as the solutions have the same weight.%\hfill\qed
    \end{proof}

%\newpage
\section{Details of the Threshold Algorithms}\label{app:ckr}

For completeness, we include a detailed description of the $\alpha$-approximation algorithm for lifted cut. This is just a collection of the results of Sharma and Vondr\'{a}k \cite{SharmaVondrak}, but understanding this is necessary for the proof of Theorem \ref{thm:lift}. The content of this section can be found in more detail in \cite{SharmaVondrak}, the only modifications we make are to perform the rounding in $k+1$ dimensions and make more clear the role of the ${(k+1)}$st vertex.

First we describe the three threshold rounding schemes: Single Threshold, Descending Thresholds, and Independent Thresholds. These are described in Algorithms \ref{alg:single}, \ref{alg:desc}, and \ref{alg:indept}, respectively. 
Each scheme is given a solution to LIFT-LP, and rounds it to an integer solution by assigning vertices to terminals.
The Single Threshold Scheme takes as input some distribution with probability density $\phi$, Descending Thresholds some distribution with density $\psi$, and Independent Thresholds with density $\xi$. Finally, these schemes are combined with appropriate parameters along with Algorithm \ref{alg:kleinberg} according to Algorithm \ref{alg:combo}, which takes additionally parameters $b,p_1,p_2,p_3,p_4\in [0,1]$, along with some probability density $\phi$.

\begin{algorithm}[ht!]
    \caption{The Single Threshold Rounding Scheme}\label{alg:single}
    \begin{algorithmic}[1]
        \State Choose threshold $\theta \in [0,1) $ with probability density $\phi(\theta)$.
        \State Choose a random permutation $\sigma$ of $[k]$.
        \ForAll{$i\in [k]$}
        \State For any unassigned $u\in V$ with $x^u_{\sigma(i)}\geq \theta$, assign $u$ to terminal $\sigma(i)$.
        \EndFor
        \State Assign all remaining unassigned vertices to terminal $k+1$
    \end{algorithmic}
\end{algorithm}

\begin{algorithm}[ht!]
    \caption{Descending Thresholds Rounding Scheme}\label{alg:desc}
    \begin{algorithmic}[1]
        \State For each $i\in [k]$, choose threshold $\theta_i \in [0,1) $ with probability density $\psi(\theta)$.
        \State Choose a random permutation $\sigma$ of $[k]$ such that $\theta_{\sigma(1)} \geq \theta_{\sigma(2)} \geq \ldots \geq \theta_{\sigma(k)} $.
        \ForAll{$i\in [k]$}
        \State For any unassigned $u\in V$ with $x^u_{\sigma(i)}\geq \theta_{\sigma_i}$, assign $u$ to terminal $\sigma(i)$.
        \EndFor
        \State Assign all remaining unassigned vertices to terminal $k+1$
    \end{algorithmic}
\end{algorithm}

\begin{algorithm}[ht!]
    \caption{Independent Threshold Rounding Scheme}\label{alg:indept}
    \begin{algorithmic}[1]
        \State For each $i\in [k]$, choose independently threshold $\theta_i \in [0,1) $ with probability density $\xi(\theta)$.
        \State Choose a uniformly random permutation $\sigma$ of $[k]$.
        \ForAll{$i\in [k]$}
        \State For any unassigned $u\in V$ with $x^u_{\sigma(i)}\geq \theta_{\sigma(i)}$, assign $u$ to terminal $\sigma(i)$.
        \EndFor
        \State Assign all remaining unassigned vertices to terminal $k+1$
    \end{algorithmic}
\end{algorithm}

\begin{algorithm}[ht!]
    \caption{The Sharma-Vondr\'ak Rounding Scheme}\label{alg:combo}
    \begin{algorithmic}[1]
        \State With probability $p_1$, choose the Kleinberg-Tardos Rounding Scheme (Algorithm \ref{alg:kleinberg}).
        \State With probability $p_2$, choose the Single Threshold Rounding Scheme (Algorithm \ref{alg:single}) with probability density $\phi$.
        \State With probability $p_3$, choose the Descending Threshold Rounding Scheme (Algorithm \ref{alg:desc}), where the thresholds are chosen uniformly in $[0,b]$.
        \State With probability $p_4$, choose the Independent Threshold Rounding Scheme (Algorithm \ref{alg:indept}), where the thresholds are chosen uniformly in $[0,b]$.
    \end{algorithmic}
\end{algorithm}

The following three Lemmas are key to the analysis of Algorithm \ref{alg:combo}. The cut density for an edge of type $(i,j)$ located at $(u_1, u_2,\ldots , u_{k+1}) \in \Delta_{k+1}$ is the limit of the probability that the given threshold scheme assigns $(u_1, u_2,\ldots , u_{k+1})$ and $(u_1+\varepsilon, u_2-\varepsilon,\ldots , u_{k+1})$ to different terminals, normalized by $\varepsilon$ as $\epsilon\to 0$.

\begin{lemma}[Lemma 5.1 in \cite{SharmaVondrak}]
Given a point $(u_1, u_2,\ldots , u_{k+1}) \in \Delta_{k+1}$ and the parameter $b$ of Algorithm \ref{alg:combo}, let $a = \frac{1-u_i-u_j}{b}$ . If $a > 0$, the cut density for an edge of type $(i, j)$, where $i\neq j$ are indices in $[k+1]$ located at $(u_1, u_2,\ldots , u_{k+1})$ under the Independent Thresholds
Rounding Scheme with parameter $b$ is at most
\begin{itemize}
    \item $\frac{2(1-e^{-a})}{ab}-\frac{(u_i+u_j)(1-(1+a)e^{-1})}{a^2b^2}$, if all the coordinates $u_1, u_2,\ldots , u_{k+1}$ are in $[0,b].$
    \item $\frac{(a+e^{-a}-1)}{a^2b}$, if $u_i\in [0,b], u_j\in (b,1]$ and $u_\ell \in [0,b]$ for all other $\ell\in [k]\setminus\{i,j\}$.
    \item $\frac{1}{b}-\frac{(u_i+u_j)}{6b^2}$, if $u_i,u_j\in [0,b]$ and $u_\ell \in (b,1]$ for some other $\ell\in [k]\setminus\{i,j\}$.
    \item $\frac{1}{3b}$, if $u_i\in [0,b], u_j\in (b,1]$ and $u_\ell \in [0,b]$ for some other $\ell\in [k]\setminus\{i,j\}$.
    \item $0$, if $u_i,u_j\in (b,1]$.
\end{itemize}
For $a = 0$, the cut density is given by the limit of the expressions above as $a \to 0$.
\end{lemma}

\begin{lemma}[Lemma 5.2 in \cite{SharmaVondrak}]
    For an edge of type $(i, j)$ located at $(u_1, u_2,\ldots , u_{k+1})$, where $i\neq j$ are indices in $[k+1]$, the cut density under the Single Threshold
Rounding Scheme is at most
\begin{itemize}
    \item $\frac{1}{2}\phi(u_i) + \phi(u_j)$, if $u_\ell \leq u_i \leq u_j $ for all other $\ell\in [k]\setminus\{i,j\}$.
    \item $\frac{1}{3}\phi(u_i) + \phi(u_j)$, if $u_i < u_\ell \leq u_j $ for some other $\ell\in [k]\setminus\{i,j\}$.
    \item $\frac{1}{2}\phi(u_i) + \phi(u_j)$, if $u_i \leq u_j < u_\ell $ for some other $\ell\in [k]\setminus\{i,j\}$.
\end{itemize}
\end{lemma}

\begin{lemma}[Lemma 5.3 in \cite{SharmaVondrak}]
    For an edge of type $(i, j)$ located at $(u_1, u_2,\ldots , u_{k+1})$, where $i\neq j$ are indices in $[k+1]$, the cut density under the Descending Thresholds
Rounding Scheme is at most
\begin{itemize}
    \item $(1-\int_{u_i}^{u_j}\psi(u)du)\psi(u_i) + \psi(u_j)$, if $u_\ell \leq u_i \leq u_j $ for all other $\ell\in [k]\setminus\{i,j\}$.
    \item $(1-\int_{u_i}^{u_j}\psi(u)du)((1-\int_{u_i}^{u_\ell}\psi(u)du))\psi(u_i) + \psi(u_j)$, if $u_i < u_\ell \leq u_j $ for some other $\ell\in [k]\setminus\{i,j\}$.
    \item $(1-\int_{u_i}^{u_j}\psi(u)du)(1-\int_{u_i}^{u_\ell}\psi(u)du)\psi(u_i) + (1-\int_{u_j}^{u_\ell}\psi(u)du)\psi(u_j)$, if $u_i \leq u_j < u_\ell $ for some other $\ell\in [k]\setminus\{i,j\}$.
\end{itemize}
\end{lemma}

The proof for each of these Lemmas is exactly as in \cite{SharmaVondrak}, save for one additional trivial observation: the cut density of an edge of type $(i,k+1)$ is at most that of an edge of type $(i,j)$ for any $j\neq i,j\neq k+1$. This is because the $(k+1)$st terminal is considered last, and has no threshold of its own, and therefore cannot increase the separation probability. With these Lemmas in hand, Theorem 5.6 of \cite{SharmaVondrak} shows, with a specific choice of parameters, that Algorithm \ref{alg:combo} is a 1.2965-approximation to \lifted as well. 

\end{document}